\documentclass[envcountsame]{llncs}

\usepackage{amsfonts} \usepackage{amsmath} \usepackage{bm}
\usepackage{latexsym} \usepackage{epsfig}
\newcommand{\ignore}[1]{}

\newcommand{\enote}[1]{} \newcommand{\knote}[1]{}
\newcommand{\rnote}[1]{}

\renewcommand{\P}[1]{{\mathbb{P}}\left[{#1}\right]}

\newcommand{\E}[1]{{\mathbb{E}}\left[{#1}\right]}
\newcommand{\Ex}[2]{{\mathbb{E}}_{#1}\left[{#2}\right]}

\renewcommand{\phi}{\varphi}

\newcommand{\copyableTheorem}[2]{
 \newtheorem*{newthm#1}{Theorem \ref{thm#1}}
 \begin{newthm#1}
   {#2}
 \end{newthm#1}
 \expandafter\newcommand\expandafter{\csname thm#1\endcsname}{
   \begin{theorem}
     \label{thm#1}
     {#2}
   \end{theorem}
 }
}

\newcommand{\uup}[2]{\mathbb{R}_{#1}\left[{#2}\right]}
\newcommand{\geo}[1]{\mathbb{G}\left[{#1}\right]}

\newcommand{\vmax}{V_{\max}}

\newenvironment{keywords}{
       \list{}{\advance\topsep by0.35cm\relax\small
       \leftmargin=1cm
       \labelwidth=0.35cm
       \listparindent=0.35cm
       \itemindent\listparindent
       \rightmargin\leftmargin}\item[\hskip\labelsep
                                     \bfseries Keywords:]}

\begin{document}
\title{Lower Bounds on Revenue of Approximately Optimal Auctions}


\author{
Balasubramanian Sivan\inst{1}\thanks{This work was done while the author was an intern
at Microsoft Research.}
\and
Vasilis Syrgkanis\inst{2} \thanks{Part of this work was done while the author was an intern
at Microsoft Research. Supported in part by ONR grant N00014-98-1-0589 and a Simons Graduate Fellowship.}
\and
Omer Tamuz\inst{3}\thanks{Part of this work was done while the author
  was a visitor at Microsoft Research. Supported in part by a Google
  Europe fellowship in Social Computing.}
}
\institute{
Computer Sciences Dept., University of Winsconsin-Madison\\ \email{balu2901@cs.wisc.edu}
\and
Dept. of Computer Science, Cornell University, Ithaca, NY, USA\\ \email{vasilis@cs.cornell.edu}
\and
Weizmann Institute, Rehovot 76100, Israel \\ \email{omert.tamuz@weizmann.ac.il}}
\maketitle
 
\begin{abstract}
We obtain revenue guarantees for the simple pricing mechanism
of a single posted price, in terms of a natural parameter of the
distribution of buyers' valuations. Our revenue guarantee applies to
the single item $n$ buyers setting, with values drawn from an
arbitrary joint distribution.  Specifically, we show that a single
price drawn from the distribution of the maximum
valuation $\vmax = \max\{V_1,V_2,\dots,V_n\}$ achieves a revenue 
of at least a $\frac{1}{e}$ fraction of the geometric expecation of
$\vmax$.  This generic bound is a measure of how revenue
improves/degrades as a function of the concentration/spread of
$\vmax$.

We further show that in absence of buyers' valuation distributions, recruiting
an additional set of identical bidders will yield a similar guarantee on
revenue. Finally, our bound also gives a measure of the extent to which one can
simultaneously approximate welfare and revenue in terms of the
concentration/spread of $\vmax$.

\end{abstract}


\section{Introduction}

 Here is a natural pricing problem: A
single item is to be sold to one among $n$ buyers.  Buyers' valuations
are drawn from some known joint distribution. How good a revenue can
be achieved by posting a single price for all the buyers, and giving
the item to the first buyer whose value exceeds the price? Can we
lower bound the revenue in terms of some properties of the
distribution?  Such a single pricing scheme is often the only option
available, for several natural reasons. In many situations, it is
illegal or not in good taste to price discriminate between buyers;
furthermore often it is not possible to implement a pricing scheme
with multiple prices.

We define the geometric expectation of a random variable before describing our
result: the geometric expectation of a random variable $X$ is given by
$e^{\E{\log(X)}}$ (see, e.g.,~\cite{paolella2006fundamental}). The geometric
expectation is always lower than the expectation, and the more concentrated the
distribution, the closer they are; indeed, the ratio between the geometric
expectation and the expectation is a natural measure of concentration around
the mean. We illustrate how the ratio of geometric and actual expectations
captures the spread of a random variable through an example in
Section~\ref{sec:mainTheorem}. 

\paragraph{Constant fraction of geometric expectation.} We show that a
single price obtains a $\frac{1}{e}$ fraction of the geometric expectation of
the maximum among the $n$ valuations $(V_1,\ldots,V_n)$, i.e. geometric
expectation of $\vmax=\max\{V_1,\ldots,V_n\}$. Thus for distributions that are
concentrated enough to have a geometric expectation of $\vmax$ that is close to
the expectation of $\vmax$, a single pricing scheme extracts a good fraction of
the social surplus. In particular, when the ratio of geometric and actual
expectations is larger than $e/4$, our revenue guarantee is larger than a
$1/4$ fraction of the welfare (and hence the optimal revenue), thus
beating the currently best known bound of $1/4$ by Hartline and
Roughgarden~\cite{HR09}. In the special case when the distribution of $\vmax$
satisfies the monotone hazard rate (MHR) property, a single price can extract a
$\frac{1}{e}$ fraction of the expected value of $\vmax$ (\cite{DRY10}).
However, since several natural distributions fail to satisfy the MHR property, 
establishing a generic revenue guarantee in terms of the
geometric expectation, and then bounding the ratio of the geometric and actual
expectation is a useful route. For instance, in Section~\ref{sec:mainTheorem} we compute
this ratio for power law distributions (which do not satisfy the MHR property) 
and show that for all exponents $m\geq 1.56$ this ratio is larger than $e/4$ thus
beating the currently known bound.

\paragraph{Why geometric expectation?}
\begin{enumerate} 
\item Since the concentration of a distribution is a crucial
property in determining what fraction of welfare (expectation of $\vmax$) can
be extracted as revenue, it is natural to develop revenue guarantees expressed
in terms of some measure of concentration.  
\item While there are several
useful measures of concentration for different contexts, in this work we
suggest that for revenue in auctions the ratio of the geometric and actual
expectations is both a generic and a useful measure --- as explained in the
previous paragraph, for some distributions our revenue guarantees are the best
known so far.  
\item The ratio of the two expectations is a
dimensionless quantity (i.e., scale free).  
\end{enumerate}

\paragraph{Second price auction with an anonymous reserve price.}
A natural corollary of the lower bound on single pricing scheme's
revenue is that the second price auction (or the Vickrey auction) with
a single anonymous reserve obtains a fraction $\frac{1}{e}$ of the
geometric expectation of $\vmax$. When buyers' distributions are
independent and satisfy a technical regularity condition, Hartline and
Roughgarden~\cite{HR09} show that the second price auction with a
single anonymous reserve price obtains a four approximation to the
optimal revenue obtainable. Here again, our result shows that for more
general settings, where bidders values could be arbitrarily
correlated, Vickrey auction with a single anonymous reserve price
guarantees a $\frac{1}{e}$ fraction of geometric expectation of
$\vmax$.

\paragraph{Second price auction with additional bidders.}
When estimating the distribution is not feasible (and hence computing
the reserve price is not feasible), a natural substitute is to recruit
extra bidders to participate in the auction to increase
competition. We show that if we recruit another set of bidders
distributed identically to the first set of $n$ bidders, and run the
second price auction on the $2n$ bidders, the expected revenue is at
least a $\frac{2}{e}$ fraction of the geometric expectation of
$\vmax$. As in the previous result, for the special case of
independent distributions that satisfy the regularity condition,
Hartline and Roughgarden~\cite{HR09} show that recruiting another set
of $n$ bidders identical to the given $n$ bidders obtains at least
half of the optimal revenue; our result gives a generic lower bound
for arbitrary joint distributions.

In the course of proving this result we also prove the following
result: in the single pricing scheme result, the optimal single price
to choose is clearly the monopoly price of the distribution of
$\vmax$. However we show that a random price drawn from the
distribution of $\vmax$ also achieves a $\frac{1}{e}$ fraction of
geometric expectation of $\vmax$.

\subsubsection{Related Work.} For the special single buyer case, Tamuz~\cite{T12}
showed that the monopoly price obtains a constant fraction of the geometric
expectation of the buyer's value. We primarily extend this result by showing that 
for the $n$ buyer setting, apart from the monopoly reserve price of $\vmax$,
a random price drawn from the distribution of $\vmax$ also gives a $\frac{1}{e}$ fraction of 
geometric expectation of $\vmax$. This is important for showing our result
by recruiting extra bidders. Daskalakis and Pierrakos~\cite{DP11} study simultaneous
approximations to welfare and revenue for settings with independent distributions that
satisfy the technical regularity condition. They show that Vickrey auction with non-anonymous
reserve prices achieves a $\frac{1}{5}$ of the optimal revenue and welfare in such settings. 
Here again, for more general settings with arbitrarily correlated values, 
our result gives a measure how the quality of such simultaneous approximations degrades with the spread of $\vmax$. 
The work of Hartline and Roughgarden~\cite{HR09} on second price auction with anonymous reserve
price / extra bidders has been discussed already.

\section{Definitions and Main Theorem}\label{sec:mainTheorem}
Consider the standard auction-theoretic problem of selling a single item among $n$ buyers. 
Each buyer  $i$ has a private (non-negative) valuation $V_i$ for receiving the item.  Buyers
are risk neutral with utility $u_i = V_ix_i - p_i$, where $x_i$ is the probability of buyer $i$ 
getting the item and $p_i$ is the price he pays. The valuation profile $(V_1,V_2,\dots,V_n)$ 
of the buyers is drawn from some arbitrary joint distribution that is known to the auctioneer. 
Let $V_{\max}=\max_i V_i$ be the random variable that
denotes the maximum value among the $n$ bidders. We denote with
$F_{\max}$ the cumulative density function of the distribution of
$V_{\max}$.

\begin{definition}
For a positive random variable $X$, the geometric expectation $\geo{X}$ 
is defined as:
$$\geo{X}=\exp(\E{\log{X}})$$
\end{definition}

We note that by Jensen's inequality $\geo{X}\leq \E{X}$ and that
equality is achieved only when $X$ is a deterministic random variable. 
Further, as noted in the introduction, the ratio of geometric and actual expectations
of a random variable is a useful measure of concentration around the mean. 
We illustrate this point through an example. 
\begin{example}
Consider the family $F_m(x) = 1-1/x^m$ of power-law
distributions for $m\geq 1$. As $m$ increases the tail of the distribution decays
faster, and thus we expect the geometric expectation to be closer to the
actual expectation.  Indeed, the geometric expectation of such a random
variable can be computed to be $e^{1/m}$ and the actual expectation to be
$\frac{m}{m-1}$.  The ratio $e^{1/m}(1-1/m)$ is an increasing function of $m$.
It reaches $1$ at $m=\infty$, i.e., when the distribution becomes a point-mass
fully concentrated at $1$. 
The special case of $m=1$ gives the equal-revenue distribution, where the
geometric expectation equals $e$ and the actual expectation is infinity.
However this infinite gap (or the zero ratio) quickly vanishes as $m$ grows; at $m=1.56$,
the ratio already crosses $e/4$ thus making our revenue guarantee better than
the current best $1/4$ of optimal revenue; at $m=4$, the ratio
already equals $0.963$.
\end{example}

For a random variable $X$ drawn from distribution $F$, define $\uup{p}{X}$ as:
$$\uup{p}{X}=p \P{X\geq p} \geq p \P{X>p}=p(1-F(p))$$
If $X$ is the valuation of a buyer, $\uup{p}{X}$ is the expected revenue obtained
by posting a price of $p$ for this buyer. Therefore $\uup{p}{V_{\max}}$ is the revenue 
of a pricing scheme that posts a single price $p$ for $n$ buyers with values $V_1,\dots,V_n$
and $\vmax = \max\{V_1,\dots,V_n\}$. 

We show that the revenue of a posted price mechanism with a single price drawn randomly
from the distribution of $\vmax$, achieves a revenue that is at least a $\frac{1}{e}$ 
fraction of the geometric expectation of $V_{\max}$, or equivalently a $\frac{1}{e}$ fraction of the 
geometric expectation of the social surplus.  

\begin{theorem}[Main Theorem]
Let $r$ be a random price drawn from the distribution of $V_{\max}$. Then: 
\begin{equation}
\Ex{r}{\uup{r}{V_{\max}}}\geq \frac{1}{e}\geo{V_{\max}}.
\end{equation}
\end{theorem}
\begin{proof}
By the definition of $\uup{r}{V}$ we have:
\begin{equation}
\Ex{r}{\uup{r}{V_{\max}}}\geq\Ex{r}{r~(1-F_{\max}(r))}.
\end{equation}
By taking logs on both the of the above equation, and using Jensen's inequality we get:
\begin{align*}
\log(\Ex{r}{\uup{r}{V_{\max}}})\geq~&\log\left(\Ex{r}{r (1-F_{\max}(r))}\right)\\
\geq~& \Ex{r}{\log(r(1-F_{\max}(r)))} \\
=~& \Ex{r}{\log(r)}+\Ex{r}{\log(1-F_{\max}(r))}.
\end{align*}
For any positive random variable $X$ drawn from a
distribution $F$ we have:
\begin{equation}
\E{\log(1-F(X))}=\int_{-\infty}^{\infty}\log(1-F(x))dF(x)=\int_{0}^{1}\log(1-y)dy=-1.
\end{equation}
So we have:
\begin{align*}
\log(\Ex{r}{\uup{r}{V_{\max}}})&\geq \Ex{r}{\log(r)}-1\\
\Ex{r}{\uup{r}{V_{\max}}}&\geq \frac{1}{e}\exp(\Ex{r}{\log(r)}=\frac{1}{e}\geo{V_{\max}}.
\end{align*}
where the last equality follows from the fact that the random reserve $r$ is drawn from $F_{\max}$. \qed
\end{proof}

Since a random price drawn from $F_{\max}$ achieves this revenue, it follows that there 
exists a deterministic price that achieves this revenue and hence the best 
deterministic price will achieve the same. 

We define the monopoly price $\eta_F$ of a distribution $F$ to be the optimal posted price in a single buyer setting when the buyer's
valuation is drawn from distribution $F$, i.e.:
\begin{equation*}
\eta_F = \arg\sup_{r} r(1-F(r))
\end{equation*}
So a direct corollary of our main theorem is the following:
\begin{corollary}
Let $\eta_{\max}$ be the monopoly price of distribution $F_{\max}$. Then:
\begin{equation*}
\uup{\eta_{\max}}{V_{\max}}\geq \frac{1}{e}\geo{V_{\max}}
\end{equation*}
\end{corollary}

\section{Applications to Approximations in Mechanisms Design}
\paragraph{Single Reserve Mechanisms for Non-iid Irregular Settings.}
A corollary of our main theorem is that in a second price auction with
a single anonymous reserve, namely a reserve drawn randomly from the distribution of $F_{\max}$ or 
a deterministic reserve of the monopoly price of $F_{\max}$, 
will achieve revenue that is a constant
approximation to the geometric expectation of the maximum value. When
the maximum value distribution is concentrated enough to have 
the geometric expectation is close to expectation it
immediately follows that an anonymous reserve mechanism's revenue is
close to that of the expected social surplus and hence the expected optimal revenue.
\begin{corollary}
  The second price auction with a single anonymous reserve achieves
  a revenue of at least $\frac{1}{e}\geo{V_{\max}}$ for arbitrarily
  correlated bidder valuations.
\end{corollary}

\paragraph{Approximation via replicating buyers in Irregular
  Settings.}
When the auctioneer is unable to estimate the distribution of $\vmax$, and therefore
unable to compute the reserve price, a well known alternative~\cite{BK96}
to achieve good revenue is to recruit additional bidders to participate
in the auction to increase competition. 
In our setting, recruiting a set of $n$ bidders
distributed identically as the initial set of $n$ bidders (i.e. following joint distribution $F$)
will simulate having a reserve drawn randomly from 
$F_{\max}$. In fact it performs even better than having a reserve --- one among 
the additionally recruited agents could be the winner and he pays the auctioneer, as against the reserve 
price setting. More formally, observe that in the setting with $2n$ bidders, half of
the revenue is achieved from the original $n$
bidders, and half from the new bidders (by symmetry). But the revenue from 
each of these parts is exactly that of the second price auction with a 
random reserve drawn from the distribution of $V_{\max}$. Hence, the
revenue of this extended second price mechanism will be twice the
revenue of a second price mechanism with a single random reserve 
drawn from the distribution of $V_{\max}$. This fact, coupled with our main theorem gives us
the following corollary.
\begin{corollary}
  The revenue of a second price auction with an additional set of bidders drawn from joint distribution $F$ is at
  least $\frac{2}{e}\geo{V_{\max}}$.
\end{corollary}

\paragraph{Approximately Optimal and Efficient Mechanisms.} Finally, we
note that when the geometric expectation of $V_{\max}$ is close to its
expectation, all our mechanisms (both the single pricing scheme, and Vickrey with a single reserve) are also approximately efficient.
\begin{corollary}
If $\geo{V_{\max}}=c \E{V_{\max}}$, a single price drawn randomly from the distribution of $F_{\max}$ 
is simultaneously $\frac{c}{e}$ approximately revenue-optimal and $\frac{c}{e}$ approximately efficient.
\end{corollary}
\begin{proof}
Since expected social welfare of a pricing scheme is at least its expected revenue, we have: 
$$\E{\text{Social Welfare}}\geq \E{\text{Revenue}}\geq \frac{1}{e}\geo{V_{\max}}\geq\frac{c}{e}\E{V_{\max}}$$
\qed
\end{proof}

\bibliographystyle{abbrv} 
\bibliography{ge}

\end{document}